\PassOptionsToPackage{hypertexnames=false}{hyperref} % get rid of warnings
\documentclass[copyright,creativecommons]{eptcs}

\usepackage[utf8]{inputenc}
\usepackage[T1]{fontenc}
\usepackage{newunicodechar}
\newunicodechar{→}{\ensuremath{\rightarrow}}
\newunicodechar{≔}{\ensuremath{:=}}
\newunicodechar{×}{\ensuremath{\times}}
\newunicodechar{⊆}{\ensuremath{\subseteq}}
\newunicodechar{⇔}{\ensuremath{\Leftrightarrow}}
\newunicodechar{¬}{\ensuremath{\neg}}
\newunicodechar{∧}{\ensuremath{\wedge}}
\newunicodechar{∨}{\ensuremath{\vee}}
\newunicodechar{⇒}{\ensuremath{\Rightarrow}}
\newunicodechar{≤}{\ensuremath{\leq}}
\newunicodechar{≥}{\ensuremath{\geq}}
\newunicodechar{∞}{\ensuremath{\infty}}
\newunicodechar{∈}{\ensuremath{\in}}
\newunicodechar{∪}{\ensuremath{\cup}}
\newunicodechar{∩}{\ensuremath{\cap}}
\newunicodechar{⋃}{\ensuremath{\bigcup}}
\newunicodechar{⋂}{\ensuremath{\bigcap}}
\newunicodechar{∅}{\ensuremath{\emptyset}}
\newunicodechar{≠}{\ensuremath{\neq}}
\newunicodechar{∉}{\ensuremath{\not\in}}
\newunicodechar{∊}{\ensuremath{\in}}
\newunicodechar{⊢}{\ensuremath{\vdash}}
\newunicodechar{ε}{\ensuremath{\epsilon}}
\newunicodechar{⌚}{\ensuremath{\clock}}
\newunicodechar{⦇}{\ensuremath{\llparenthesis}}
\newunicodechar{⦈}{\ensuremath{\rrparenthesis}}
\newunicodechar{∼}{\~{}}
\newunicodechar{ℕ}{\ensuremath{\mathbb{N}}}
\newunicodechar{ℤ}{\ensuremath{\mathbb{Z}}}
\newunicodechar{ℙ}{\ensuremath{\mathbb{P}}}
\newunicodechar{⊤}{\ensuremath{\top}}
\newunicodechar{⊥}{\ensuremath{\bot}}
\newunicodechar{⋅}{\ensuremath{\cdot}}
\newunicodechar{〈}{\ensuremath{\langle}}
\newunicodechar{⟨}{\ensuremath{\langle}}
\newunicodechar{〈}{\ensuremath{\langle}}
\newunicodechar{〉}{\ensuremath{\rangle}}
\newunicodechar{⟩}{\ensuremath{\rangle}}
\newunicodechar{〉}{\ensuremath{\rangle}}
\newunicodechar{∑}{\ensuremath{\sum}}
\newunicodechar{Σ}{\ensuremath{\Sigma}}
\newunicodechar{∏}{\ensuremath{\prod}}
\newunicodechar{Π}{\ensuremath{\Pi}}
\newunicodechar{∀}{\ensuremath{\forall}}
\newunicodechar{∃}{\ensuremath{\exists}}

\providecommand{\event}{ThEdu'17 Post-Proceedings} % Name of the event you are submitting to
\usepackage{underscore}           % Only needed if you use pdflatex.

\usepackage{csquotes}
\usepackage{verbatim}
\usepackage[nounderscore]{syntax}
\usepackage{graphicx}
\usepackage{amsmath}
\usepackage{amsthm}
\usepackage{amssymb}
\usepackage{mathbbol}
\DeclareSymbolFontAlphabet{\mathbb}{AMSb}

\usepackage{algorithm}
\usepackage[noend]{algpseudocode}

\makeatletter
\def\verbatim@font{\scriptsize\ttfamily}
\makeatother
\newcommand{\button}[1]{\raisebox{-0.1\baselineskip}
{\includegraphics[height=0.9\baselineskip]{#1.png}}}

\newenvironment{xverbatim}{\quote\verbatim}{\endverbatim\endquote}

\newcommand{\ldbrack}{\Lbrack\,}
\newcommand{\rdbrack}{\,\Rbrack}
\newcommand{\semantics}[1]{{\ensuremath\ldbrack{#1}\rdbrack}}

\newtheorem{theorem}{Theorem}
\newtheorem{definition}{Definition}

\title{Teaching the Formalization of Mathematical Theories and Algorithms 
via the Automatic Checking of Finite Models\thanks{Supported by
the Johannes Kepler University Linz, Linz Institute of Technology (LIT), 
Project LOGTECHEDU \enquote{Logic Technology for Computer Science Education}.}}

\author{Wolfgang Schreiner \qquad Alexander Brunhuemer \qquad Christoph Fürst
\institute{Research Institute for Symbolic Computation (RISC) \&
Linz Institute of Technology (LIT) \\
Johannes Kepler University, Linz, Austria}
\email{Wolfgang.Schreiner@risc.jku.at
\quad alex.brunhuemer@gmail.com
\quad Christoph.Fuerst@risc.jku.at}}

\def\titlerunning{Teaching Formalization via Automatic Checking}
\def\authorrunning{W. Schreiner, A. Brunhuemer \& C. Fürst}
\begin{document}
\maketitle

\begin{abstract}

Education in the practical applications of logic and proving such as the formal
specification and verification of computer programs is substantially hampered by
the fact that most time and effort that is invested in proving is actually
wasted in vain: because of errors in the specifications respectively algorithms
that students have developed, their proof attempts are often pointless (because
the proposition proved is actually not of interest) or a priori doomed to fail
(because the proposition to be proved does actually not hold); this is a
frequent source of frustration and gives formal methods a bad reputation. RISCAL
(RISC Algorithm Language) is a formal specification language and associated
software system that attempts to overcome this problem by making logic
formalization fun rather than a burden. To this end, RISCAL allows students to
easily validate the correctness of instances of propositions respectively
algorithms by automatically evaluating/executing and checking them on (small)
finite models. Thus many/most errors can be quickly detected and subsequent
proof attempts can be focused on propositions that are more/most likely to be
both meaningful and true.

\end{abstract}

\section{Introduction}
\label{sect:introduction}

From a student's perspective, education in the practical applications of logic
and proving (such as the formal specification of a computational problem, the
development of an algorithm that is supposed to solve this problem, and the
formal verification that the algorithm indeed implements the specification) is
often a source of frustration: on the one hand, the specification that she
writes may be too weak (sometimes even trivially satisfied) such that a
successful verification may be of little value (sometimes even completely
pointless); on the other hand, she may waste a lot of time in proof attempts
that are a priori doomed to fail due to a variety of possible
\enquote{show-stoppers}: her specification may be too strong (sometimes even
not implementable), her algorithm may not implement the specification, and the
additional information that she has to provide for the derivation of
verification conditions (in particular annotations of loops by invariants) may
be not adequate (invariants may be too strong or too weak). To the student it
would thus be very re-assuring, if formal definitions, specifications,
algorithms, and annotations could be (quickly) validated to find out apparent errors
\emph{before} starting any (costly) proof attempts.

Of course, to establish the truth of a proposition interpreted over an infinite
domain generally requires a symbolic proof (which cannot be fully automated,
thus most program verification environments make use of interactive proof
assistants rather than just relying on automated provers); propositions over
finite domains, however, can be automatically checked without proof by
systematically enumerating all possible values for the quantified variables of a
formula. The problem then, however, is that propositions over finite domains are
not necessarily semantically connected to corresponding formulas over infinite
domains; for instance, the formula $(\forall x.\ \exists y.\ y > x)$ is true when
interpreted over the infinite set $\mathbb{N}$ of natural numbers but false when
interpreted over every finite subset of $\mathbb{N}$.

To overcome this problem, we may restrict all variables to finite types whose
sizes are bounded by some parameter $n\in\mathbb{N}$ (or multiple such
parameters); the specification is thus interpreted over a domain that is finite
but of arbitrary size. It may then involve a formula $(\forall
x\in\mathbb{N}[n].\ F)$ where $\mathbb{N}[n]$ denotes all natural numbers less
than equal $n$. If we instantiate $n$ with, say, $5$, we get a formula instance
$(\forall x\in\mathbb{N}[5].\ F)$ that can be effectively evaluated; thus all
specifications and annotations of programs can be effectively checked during the
execution of the program (\emph{runtime assertion checking}). Furthermore, since
also the domains of program variables are correspondingly restricted, we can
effectively execute a program and check its annotations for all possible inputs
(\emph{model checking}); only if we do not find errors, the verification of the
general specification may be attempted, e.g. the proof of $(\forall
n\in\mathbb{N}.\ \forall x\in\mathbb{N}[n].\ F)$.

Based on this idea and prior expertise with computer-supported program
verification especially in educational scenarios~\cite{Schreiner2008,Schreiner2012},
we have developed the specification language RISCAL (RISC Algorithm
Language) and associated software system~\cite{RISCAL,Schreiner2017}. 
RISCAL has been designed in such a way that
\begin{itemize}
\item every type has (arbitrarily many but) only finitely many values, and thus
\item every language construct is executable, and thus
\item every constant, function, predicate, theorem, procedure can be evaluated.
\end{itemize}
While every RISCAL type such as $\mathbb{N}[n]$ is finite, it may depend on a
constant $n$ not defined in the specification; thus the specification denotes an
infinite class of models of which every instance (corresponding to every
concrete values of $n$) is finite and executable. With RISCAL we may thus
validate some model instances before attempting to prove the correctness of all
models (for arbitrary values of $n$).

RISCAL is intended to model, rather than low-level program code, high-level
algorithms as can be found in textbooks on discrete mathematics. It thus
provides a richer collection of built-in data types (e.g., sets and maps) and
operations (e.g., quantified formulas as program conditions and implicit
definitions of values by formulas) than can be typically found in real
programming languages; in particular, RISCAL also supports various
non-deterministic phrases based on the choice of a value with a determining
property. This enables the implicit definition of functions respectively of
non-deterministic algorithms that have not necessarily a uniquely defined
result. The current version of the RISCAL software supports model checking of
formulas, specifications, and algorithms via the runtime assertion checking of
all possible executions, based on the executability of all specifications and
annotations (further automatic mechanisms based on SMT solving and interactive
proofs with the help of proving assistants will be added in the future). The
implementation allows to lazily evaluate/execute all possible
evaluation/execution paths in non-deterministic expressions/statements and have
theorems and algorithms checked in each path.

RISCAL is related to a large body of prior research; we only give a short
account of the work that seems most relevant, mainly focusing on approaches that
allow students to validate formulas, specifications, algorithms not only by
symbolic proving. For instance, the classic software system Tarski's
World~\cite{Barker2008} has applied a visual approach: it demonstrates the
semantics of first-order logic through games in which three-dimensional worlds
are populated with various geometric figures that test the truth of logic
formulas; however, this framework is oriented towards beginners and has limited
expressiveness.

Various languages of automated reasoning systems have some executable flavor,
which allows to evaluate formal definitions, e.g., the formal proof management
system Coq~\cite{Bertot2004}, the generic proof assistant
Isabelle~\cite{Nipkow2017} (which has been, e.g., used to define the formal
semantics of a simple imperative programming language in executable form~\cite{Nipkow2014}),
or the system Theorema~\cite{Buchberger2016} for computer supported mathematical
theorem proving (which considers computing as a special form of proving). As for
algorithm languages, Alloy~\cite{Jackson2011} is a language for describing
structures and their relationships based on a relational logic; the Alloy
Analyzer is a satisfiability solver that finds structures satisfying certain
constraints. The formal method Event-B and the supporting Rodin
system~\cite{Abrial2010} have been developed for the modeling and analysis of
systems, based on set theory as a modeling notation and the concept of
refinement to represent systems at different abstraction levels. The
specification language TLA+ for describing concurrent systems~\cite{Lamport2002}
and the corresponding algorithm language PlusCal are supported by the TLC model
checker. Also the Vienna Development Method VDM~\cite{Larsen2016} provides an
expressive language for modeling algorithms and supports by its software
Overture the testing of specifications. As for real programming languages with
support for formal specification, the programming language WhyML of the
verification environment Why3~\cite{Filliatre2013} and Microsoft's programming
language Dafny~\cite{Leino2010} provide a rich variety of built-in specification
constructs; however, both WhyML and Dafny do not support model checking. Also
for industrially supported languages, in particular around the Java Modeling
Language (JML) for the formal specification of Java programs, an ecosystem of
supporting tools has been developed~\cite{Burdy2005}, including runtime
assertion checkers, extended static checkers, and full-fledged verifiers.

The thesis~\cite{Ritirc2016} has evaluated various of the languages and tools
mentioned above with respect to their suitability for specifying and verifying
mathematical algorithms; ultimately it favored, with some reservations, PlusCal
and VDM. PlusCal is attractive because of its roots in first-order logic and set
theory. However, by its heritage from TLA+, PlusCal has no static type system,
no possibility to directly specify algorithm contracts and annotations, and it
does not support recursion; also models are not necessarily restricted to finite
size such that model checks may not terminate. VDM/Overture mainly aims at the
development and analysis of models of complex systems and the generation of
executable code from specifications, but it is generally also suitable for the
purpose addressed by RISCAL. However, because its type system is based on
infinite types, it does not really support exhaustive model checking but only
(combinatorial) testing where the developer has to explicitly specify sets of
executions to be performed and annotations be checked; also there is no direct
support for the formulation/checking of mathematical theorems. Furthermore,
execution is deterministic such that e.g. in non-deterministic choices always
the first values (or optionally a random value) is chosen. Quantified constructs
are executable only if the bound variables take values from user-defined finite
collections (sets or sequences), and it is not possible to directly annotate
loops with invariants (only global system/type invariants are supported).
 
The development of RISCAL has been triggered and shaped by above findings.
Compared with the approaches mentioned, only RISCAL supports a rich language for
conveniently defining mathematical theories and specifying, modeling, and
annotating (potentially non-deterministic) algorithms that is both fully
checkable and semantically linked to arbitrarily sized models. Furthermore, we
are convinced that the RISCAL software is by far the most streamlined and
easiest to use for our purposes. So far, RISCAL has been used to formalize
algorithms in number theory~\cite{Fuerst2017}, discrete
mathematics~\cite{Brunhuemer2017}, computer algebra, and logic. It is currently
applied in a course on \enquote{Formal Methods in Software Development} for
master programmes in computer science and mathematics. The ultimate goal is to
build up a comprehensive library and accompanying lecture materials;
furthermore, due to the immediate feedback of the system about the correctness
of theorems and algorithms, we envision RISCAL as an ideal vehicle for
self-directed and self-paced learning in STEM education. Students have already
given very positive feedback how natural the use of the system feels and how
easy it is to play with formal specifications.

The remainder of this paper is structured as follows: Section~\ref{sect:RISCAL}
gives an overview on the RISCAL specification language and associated software
system. Section~\ref{sect:education} outlines the envisioned strategy of
applying RISCAL in education on the formalization of theories and algorithms.
Section~\ref{sect:examples} describes first attempts on the development of
formal specification libraries in various mathematical domains.
Section~\ref{sect:conclusions} concludes and discusses our plans on the
further development and application of RISCAL.

\section{The RISCAL Language and System}
\label{sect:RISCAL}

\begin{figure}
\centering
\includegraphics[width=0.72\textwidth]{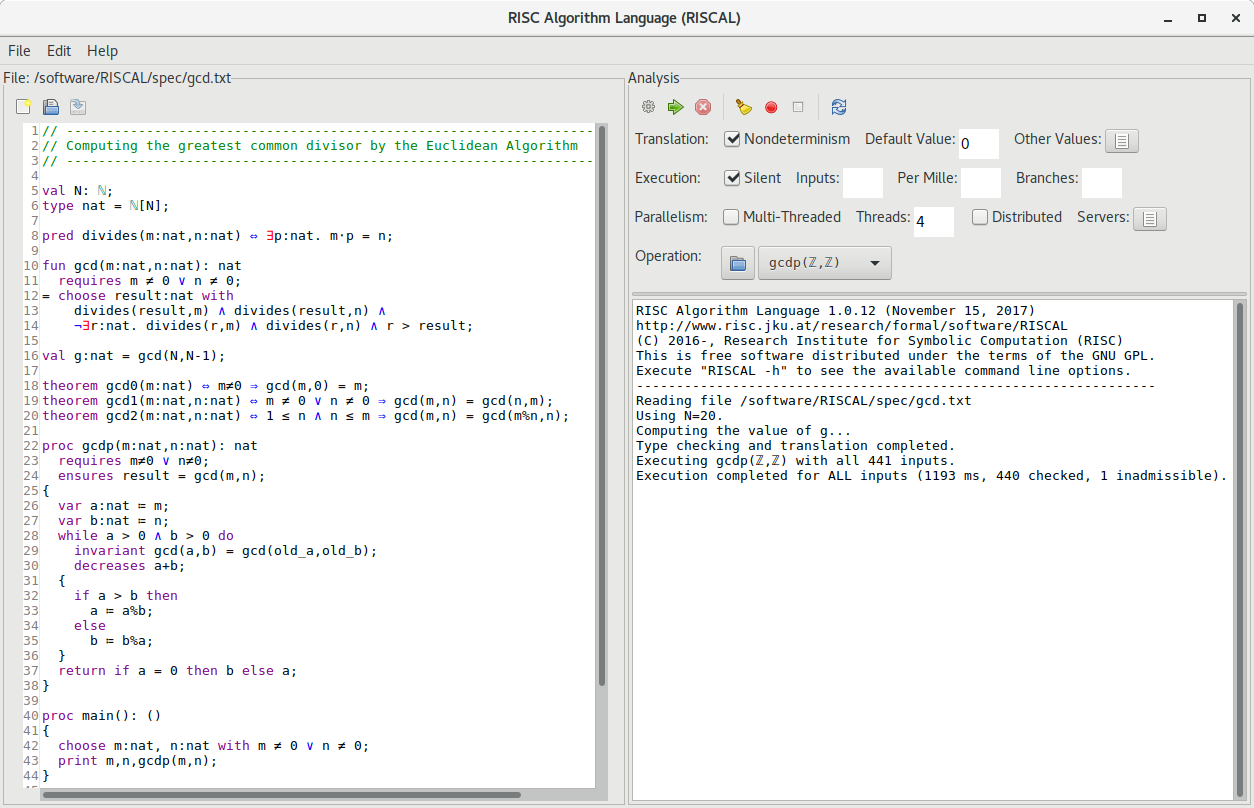}
\caption{The RISCAL System}
\label{fig:RISCAL}
\end{figure}

In this section, we give a short account of the RISCAL language and
software system; for more details, see the tutorial
and reference manual~\cite{Schreiner2017}.

\paragraph{System}

The user interface of the RISCAL software system is depicted in
Figure~\ref{fig:RISCAL}; it contains an editor frame for RISCAL specifications
on the left and the control widgets and output frame of the checker on the
right. The RISCAL specification language is based on a statically typed variant
of first order predicate logic. On the one hand, it allows to develop
mathematical theories such as that of the greatest common divisor depicted at
the top of Figure~\ref{fig:gcd}; on the other hand, it also enables the
specification of algorithms such as Euclid's algorithm depicted in the same
figure below (theory and specification will be discussed later). The lexical
syntax of the language includes Unicode characters for common mathematical
symbols; these may be entered in the RISCAL editor via ASCII shortcuts; e.g.,
the character $∀$ is entered by first typing the text \texttt{forall} and then
pressing the keys \texttt{<Ctrl>} and \texttt{\#} simultaneously.

\begin{figure}[t]
\begin{xverbatim}
val N: ℕ; type nat = ℕ[N];

pred divides(m:nat,n:nat) ⇔ ∃p:nat. m⋅p = n;
fun gcd(m:nat,n:nat): nat
  requires m ≠ 0 ∨ n ≠ 0;
= choose result:nat with
    divides(result,m) ∧ divides(result,n) ∧
    ¬∃r:nat. divides(r,m) ∧ divides(r,n) ∧ r > result;

theorem gcd0(m:nat) ⇔ m≠0 ⇒ gcd(m,0) = m;
theorem gcd1(m:nat,n:nat) ⇔ m ≠ 0 ∨ n ≠ 0 ⇒ gcd(m,n) = gcd(n,m);
theorem gcd2(m:nat,n:nat) ⇔ 1 ≤ n ∧ n ≤ m ⇒ gcd(m,n) = gcd(m%n,n);

proc gcdp(m:nat,n:nat): nat
  requires m≠0 ∨ n≠0;
  ensures result = gcd(m,n);
{
  var a:nat ≔ m; var b:nat ≔ n;
  while a > 0 ∧ b > 0 do
    invariant gcd(a,b) = gcd(old_a,old_b);
    decreases a+b;
  {
    if a > b then a ≔ a%b; else b ≔ b%a;
  }
  return if a = 0 then b else a;
}
\end{xverbatim}
\caption{Euclid's Algorithm in RISCAL}
\label{fig:gcd}
\end{figure}

\paragraph{Language}

RISCAL specifications consist of declarations of 
the following kinds of entities:

\begin{description}

\item[Types] \texttt{type $I$ = $T$} introduces a named type $I$ defined by the
type expression $T$; types include booleans, integers, sets, tuples, records,
arrays, and maps (partial functions). All types are \emph{finite}, e.g., for
integer constants $N$ and $M$ with $N\leq M$ the type \texttt{ℤ[$N$,$M$]}
denotes the type of all integers in the interval $[N,M]$ while
\texttt{Array[$N$,$T$]} denotes the type of all arrays of length $N \geq 0$ with
elements of type $T$. Recursive (algebraic) data types (whose values are terms
of a finite depth $N \geq 0$) may be introduced by a declaration
\texttt{rectype($N$) $T$ = $c$($T_1$,\ldots,$T_n$) | \ldots} .

\item[Constants] \texttt{val $I$:$ℕ$} introduces an unspecified natural number
constant $I$ while \texttt{val $I$:$T$ = $E$} introduces a constant $I$ of type
$T$ which is explicitly defined by a term $E$. Terms can be composed from a rich
variety of built-in functions and quantifiers, e.g. the term
\texttt{(∑x:ℕ[N]~with~x\%2≠0.~x⋅x)} denotes the sum of the squares of all odd
natural numbers less than equal $N$.

\item[Functions and Predicates] \texttt{fun $I$(\ldots):$T$ = $E$} introduces a
function $I$ with result type $T$; the result value is defined by expression $E$
of type $T$; correspondingly, \texttt{pred $I$(\ldots) ⇔ $F$} defines a
predicate $I$, i.e., a Boolean-valued function whose result is defined by
formula $F$ (an expression of type \texttt{Bool}). Formulas can be written in a
notation that is close to typical mathematical practice, e.g.,
\texttt{(∀x:ℕ[N],y:ℕ[N].~x ≤ y ⇒ ∃z:Nat.~z ≤ y ∧ x+z = y)} is such a formula.

% which can be also written as \texttt{(∀x:ℕ[N],y:ℕ[N] with x ≤ y.~∃z:Nat with z ≤
% y.~x+z = y)}.

\item[Theorems] \texttt{theorem $I$ ⇔ $F$} introduces a Boolean constant $I$
whose value is defined by a formula $F$; this declaration asserts that the value
of $I$ is \texttt{true}, i.e., that $F$ is valid. Likewise \texttt{theorem
$I$(\ldots) ⇔ $F$} introduces a predicate $I$ defined by formula $F$; this
declaration asserts that $F$ is valid for all
possible parameter values, i.e., $F$ is implicitly universally quantified.

\item[Procedures] \texttt{proc $I$(\ldots):$T$ \{ $C$; return $E$; \}}
introduces a procedure $I$ with result type $T$. A procedure is a function whose
result value is determined by the execution of command $C$; this establishes a
context (determined by the values of modifiable variables in the procedure) in
which the value of expression $E$ is evaluated to denote the result value.
Commands support the usual algorithmic constructs like variable assignments,
command sequences, and various forms of conditionals and loops. Loops may be
annotated by \texttt{invariant $F$} to indicate that formula $F$ is true before
and after every iteration of the loop; the annotation \texttt{decreases $E$}
indicates that the value of the termination measure $E$ (a natural number
expression) is decreased in every iteration. A command \texttt{assert $F$}
indicates that formula $F$ is \texttt{true} when the command is executed.

\end{description}

Parameterized entities may be annotated by preconditions of form
\texttt{requires $F$} which states that only those parameter values are legal
that satisfy formula $F$; correspondingly annotation \texttt{ensures $F$} states
that only a result (denoted by the keyword \texttt{result}) is legal that
satisfies $F$. These entities may be also defined recursively; by an annotation
\texttt{decreases $E$} a termination measure $E$ is stated, i.e., an expression
$E$ that evaluates to a natural number which is decreased in every recursive
invocation.

The algorithmic language RISCAL also supports \emph{non-deterministic}
expression evaluations respectively command executions, which may considerably
simplify the formulation of many algorithms. For instance, the term
\texttt{(choose $I$:$T$~with $F$)} denotes some value~$I$ of type $T$ that
satisfies formula $F$; if no such value exists, the value of the term is
undefined. A corresponding command introduces a constant~$I$ with that property
into the current context. The conditional command \texttt{choose \ldots\ then
$C_1$ else $C_2$} executes command~$C_1$, if such a constant exists, and $C_2$,
otherwise; the loop \texttt{choose \ldots\ do $C$} performs the choice
repeatedly and terminates when no more choice is possible. The loop \texttt{for
\ldots\ do $C$} executes the body $C$ for all possible choices in an unspecified
order.

\paragraph{Example}

The specification listed in Figure~\ref{fig:gcd} introduces the mathematical
theory of the greatest common divisor and its computation by the Euclidean
algorithm; this theory is restricted to the domain of all natural numbers less
than equal \texttt{$N$}:

\begin{itemize}

\item The theory first introduces the undefined constant $N$ which is
then used as the domain bound of the type \texttt{ℕ[$N$]}
subsequently called \texttt{nat}.

\item It then defines a predicate \texttt{divides($m$,$n$)} which denotes $m|n$
($m$ divides $n$) and subsequently a function \texttt{gcd(m,n)} which denotes
the greatest common divisor of $m$ and $n$. This function is introduced by an
implicit definition: for any $m,n$ with $m\neq 0$ or $n \neq 0$, its result is
the largest value $\mathit{result}$ that divides both $m$ and $n$.

\item The theorems \texttt{gcd0($m$)}, \texttt{gcd1($m$,$n$)}, and
\texttt{gcd2($m$,$n$)} describe the essential mathematical propositions on
which the correctness proof of the Euclidean algorithm is based.

\item The procedure \texttt{gcdp($m$,$n$)} embeds an iterative implementation of
the Euclidean algorithm. Its contract specified by the clauses \texttt{requires}
and \texttt{ensures} states that the procedure behaves exactly as the implicitly
defined function; the loop annotations \texttt{invariant} and \texttt{decreases}
describe essential knowledge for proving the total correctness of the procedure
(here the automatically introduced constants \textit{old\_a} and \textit{old\_b}
denote the values of the program variables $a$ and $b$ before entering the loop,
i.e., in this context, $m$ and $n$, respectively).

\end{itemize}

\paragraph{Evaluation, Execution, and Checking}

Whenever the user saves a specification, it is automatically syntactically and
semantically processed, i.e., parsed, type-checked, and translated into an
executable internal representation (see below for more details on the
implementation). Errors are reported by graphical markers in the editor frame
and by textual messages in the output frame.

For the semantic processing, all globally defined constants are immediately
evaluated (by evaluating the defining terms/formulas which may only refer to
already previously processed entities). Those natural number constants whose
values have not been defined in the specification receive their values from the
current system settings that the user may control in the graphical interface: by
pressing the button \enquote{Other Values} a menu pops up that allows to give
values to selected constants; if a constant is not given a value here, the
\enquote{Default Value} from the input box in the main window is chosen.
Since thus all constants receive specific values, all types depending on
these constants receive an interpretation as specific finite sets of values
(which are internally implemented as lazily evaluated sequences).

The translation proceeds according to either a deterministic or a
non-deterministic model of expression evaluation respectively command execution:

\begin{itemize} 

\item In the deterministic model, every non-deterministic choice results in only
\emph{one} value; if no value can be chosen (because of an unsatisfiable side condition
$F$ specified in the choice), the program aborts with an error message.

\item In the non-deterministic model, every non-deterministic choice
(ultimately) results in \emph{all} possible values; since all types are finite,
also the number of choices is finite, such that all expressions can be
evaluated in a finite amount of time.

\end{itemize}

The non-deterministic model is implemented by translating every expression that
may have non-deterministic semantics to a lazily evaluated sequence of values;
the evaluation of an expression respectively execution of a procedure first
proceeds according to whatever value is first delivered by all streams; after
that execution, it \enquote{backtracks} to the last stream and then proceeds
with the value that is delivered next; if this stream has delivered all values,
execution backtracks to the previous stream, and so on. Thus ultimately the
complete \enquote{tree} of all possible choices is processed in a depth-first
fashion. However, since in general the non-deterministic model requires the
processing of exponentially many evaluation/execution paths, the deterministic
mode is the default; the non-deterministic mode is only applied, if the user
has explicitly checked the selection box \enquote{Nondeterminism}.

For all parameterized entities (functions, predicates, theorems, procedures) the
menu \enquote{Operation} allows to select the entity; by pressing the
\enquote{Run} button \button{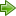}, the system generates (in a lazy fashion)
all possible combinations of parameter values that satisfy the precondition of
the operation, invokes the operation on each of these, and prints the
corresponding result values. If the selection box \enquote{Silent} is checked,
the output for each operation is suppressed; however, each execution still
checks the correctness of all annotations (preconditions, postconditions,
theorems, invariants, termination measures, and assertions). If the checking
thus completes without errors, we have validated that the operation satisfies
the specification for the domains determined by the current choices of the
domain bounds.

\paragraph{Example}

We may check the specification listed in Figure~\ref{fig:gcd} in various ways;
for this, in the following the value $N=20$ is used.

First we execute \texttt{gcd} in nondeterministic mode to validate that
the specification of the function allows for every admissible input one and only
one output and that this output is indeed the expected one:
\begin{xverbatim}
Executing gcd(ℤ,ℤ) with all 441 inputs.
Ignoring inadmissible inputs...
Branch 0:1 of nondeterministic function gcd(1,0):
Result (0 ms): 1
Branch 1:1 of nondeterministic function gcd(1,0):
No more results (4 ms).
...
Branch 0:440 of nondeterministic function gcd(20,20):
Result (1 ms): 20
Branch 1:440 of nondeterministic function gcd(20,20):
No more results (7 ms).
Execution completed for ALL inputs (5187 ms, 440 checked, 1 inadmissible).
\end{xverbatim}
By checking the option \texttt{Silent} we see that that the execution is
actually pretty fast (unchecking the option \texttt{Nondeterministic} would
speedup it further by a factor of more than two):
\begin{xverbatim}
Executing gcd(ℤ,ℤ) with all 441 inputs.
Execution completed for ALL inputs (273 ms, 440 checked, 1 inadmissible).
\end{xverbatim}
Likewise, checking the theorems proceeds very quickly, e.g.:
\begin{xverbatim}
Executing gcd2(ℤ,ℤ) with all 441 inputs.
Execution completed for ALL inputs (256 ms, 441 checked, 0 inadmissible).
\end{xverbatim}
Similarly, we can validate that the procedure satisfies its specification:
\begin{xverbatim}
Executing gcdp(ℤ,ℤ) with all 441 inputs.
Execution completed for ALL inputs (933 ms, 440 checked, 1 inadmissible).
\end{xverbatim}
This check indeed evaluates the procedure specification and the embedded loop
annotations; if we introduce an error, e.g. by modifying the last line of the
procedure to
\begin{xverbatim}
return if a = 0 then 0 else a;
\end{xverbatim}
the error is immediately detected:
\begin{xverbatim}
Executing gcdp(ℤ,ℤ) with all 441 inputs.
ERROR in execution of gcdp(0,1): evaluation of
  ensures result = gcd(m, n);
at line 24 in file gcd.txt:
  postcondition is violated by result 0
ERROR encountered in execution.
\end{xverbatim}
More uses of the RISCAL checker will be shown in the following sections.

\paragraph{Implementation}

RISCAL has been implemented in Java (using the Eclipse Standard Widget Toolkit
SWT for its graphical user interface). The executable internal representation of
a specification is essentially a Java version of a denotational semantics of the
specification, implemented on top of the lambda expressions introduced in Java
8.

For instance, in the deterministic model, the semantics of a command is a
function from contexts (variable bindings) to contexts. In the non-deterministic
model, it is a function from contexts to potentially infinite sequences of
contexts; such as sequence is modeled as a function that either returns
\enquote{null} denoting the end of the sequence or a pair of a value and another
sequence. This framework is expressed by the following mathematical domain definitions:
\begin{align*}
& \mathit{ComSem} := \mathit{Single} + \mathit{Multiple} \\
& \mathit{Single} := \mathit{Command} \to
(\mathit{Context}\to\mathit{Context}) \\
& \mathit{Multiple} := \mathit{Command} \to
(\mathit{Context}\to\mathit{Seq}(\mathit{Context})) \\
& \mathit{Seq}(T) := 
\mathit{Unit}\to\left(\mathit{Null} +\mathit{Next}(T,\mathit{Seq}(T))\right)
\end{align*}
Here $\mathit{ComSem}$ represents the (deterministic or non-deterministic)
semantics of commands and $\mathit{Seq}(T)$ represents the domain of sequences
of type $T$. The deterministic semantics of a one-sided conditional
command can be then defined by the following value of type $\mathit{Single}$:
\begin{align*}
\semantics{\texttt{if}\ E\ \texttt{then}\ C} :=
\lambda c. \textrm{if}\ \semantics{E}(c)\ \textrm{then}\ \semantics{C}(c)\ \textrm{else}\ c
\end{align*}
Corresponding Java 8 versions of these types are the following:
\begin{xverbatim}
public static interface ComSem  {
  public interface Single extends ComSem, Function<Context,Context> { }
  public interface Multiple extends ComSem, Function<Context,Seq<Context>> { }
}
public interface Seq<T> extends Supplier<Seq.Next<T>> { 
  // public Seq.Next<T> get();
  public final static class Next<T> {
    public final T head; public final Seq<T> tail;
    ...
  }
}
\end{xverbatim}
Here the deterministic semantics of the conditional command can be
defined by the following function:
\begin{xverbatim}
static ComSem.Single ifThenElse(BoolExpSem.Single E, ComSem.Single C) 
{ return (Context c) -> E.apply(c) ? C.apply(c) : c; }
\end{xverbatim}
In a similar style also the non-deterministic semantics can be defined
and implemented based on higher-order functions for combining functions
on sequences of values. 

The model checker applies the executable representation of an operation to all
values of the domain of the operation; this is based on a translation of types
to (again lazily evaluated) sequences of values such that from the interface of
a function the sequence of all possible arguments can be generated.

For speeding up the checking of larger models, both a multi-threaded and a
distributed version of the checker have been implemented. The multi-threaded
version can be selected by the check box \enquote{Multi-Threaded} where by the
input box \enquote{Threads} the number of worker threads can be selected. These
threads run on the local computer and iteratively request from the main thread
new inputs to which the chosen operation is to be applied; thus the domain of
inputs is processed by all threads in parallel. Additionally or alternatively
the distributed version can be selected by the check box \enquote{Distributed}
where by the button \enquote{Servers} connections to one or more remote
servers can be established. On each server an instance of RISCAL is started as a
separate process to which the main process forwards the specification for a
local translation into the executable form; each server process then repeatedly
queries from the main process a range of inputs which can be processed by
multiple threads per server (in addition to the threads running on the main
process). For more details, consult the reference manual~\cite{Schreiner2017}.

\paragraph{Proof-Based Verification}

Since RISCAL is based on checking rather than deriving and proving verification
conditions, there is not any direct connection of RISCAL to a verification
calculus like Hoare's axiomatic system or Dijkstra's predicate transformers.
However, as will be sketched in the following section, we plan as future work to
integrate also such a calculus as a way of validating by checking that loop
invariants are strong enough to carry through proof-based verification over
infinite domains.

\section{Using RISCAL in Education}
\label{sect:education}

RISCAL shall support the education in formal logic with emphasis on the formal
specification and verification of programs respectively algorithms (abstract
programs). The particular goal of RISCAL is to give students \emph{immediate
feedback} about the interpretation and adequacy of formal definitions and
specifications \emph{before} they attempt to formally prove theorems such as
verification conditions for the correctness of programs. The environment shall
thus encourage and support \emph{self-paced instruction} where students
\enquote{play} with multiple variations of definitions, specifications, and
annotations, and by the feedback of the system learn to interpret their meaning,
investigate their properties, and judge their appropriateness for the intended
purpose. The goal is to rule out subtle errors in definitions, theorems,
specifications, and annotations which may make subsequent proofs pointless
(since the theorem proved does not capture the informal intention) and/or
impossible (since the theorem does actually not hold); thus the major sources of
frustration in dealing with formal methods can be avoided or at least minimized.
Only when by these activities the adequacy of the formal artifacts has been
satisfactorily validated, formal proofs shall be attempted, which by the
previous activities have a high (at least much higher) chance of being both
meaningful and successful.

In particular, RISCAL can support the following activities:

\begin{enumerate}

\item \emph{The formalization of theories:} this involves the definition
of types, constants, functions, and predicates and the formulation of theorems
that claim certain properties of these theories. Subsequently it has to be
validated that these definitions indeed capture the intentions of the human
respectively that the claimed properties are indeed correct.

RISCAL can support this process by evaluating the definitions of functions and
predicates for all (or also just some) possible input values and observing their
outputs. Moreover, RISCAL may evaluate the theorems for all possible input
values (i.e., values for the universally quantified variables of the theorem)
and check their correctness; if a theorem is violated, the system reports
witnesses of the violation (i.e., values for the variables that make the
defining formula false). Additionally, the user may annotate every expression
$E$ as \texttt{print $E$} such that its evaluation prints the result value as a
side-effect; thus the evaluation of terms and formulas can be traced.

\item \emph{The specification of algorithms:} this involves the definition of
pre- and post-conditions of envisioned algorithms. Subsequently it has to be
validated that these specifications satisfy certain criteria and indeed describe
the intended input/output behavior of the algorithm. In particular, RISCAL may
validate precondition $P(x)$ on input $x$ and postcondition $Q(x,y)$ on input~$x$
and output~$y$ by the following activities (for simplicity, we write the
following formulas in common mathematical notation, concrete RISCAL
counterparts will be subsequently shown):

\begin{enumerate}

\item Check the validity of $\exists x. P(x)$ which verifies that the
precondition is satisfiable. If this theorem does not hold, the specification is
trivial: a proof of correctness of an algorithm with respect to the
specification typically quickly succeeds but is pointless (indeed small logical
errors may lead to the definitions of preconditions that are equivalent to
\enquote{false}).

\item Check the validity of $\forall x. P(x) \Rightarrow \exists y. \neg Q(x,y)$
which verifies that the postcondition is not valid, i.e., not satisfied by every
output (if some inputs actually allow arbitrary outputs, alternatively the
weaker theorem $\exists x,y. P(x) \wedge \neg Q(x,y)$ may be checked). If none
of these theorems hold, the specification is trivial: a proof of correctness of
an algorithm with respect to the specification typically quickly succeeds but is
pointless (indeed small logical errors may lead to the definitions of
postconditions that are equivalent to \enquote{true}).

\item Check the validity of $\forall x. P(x) \Rightarrow \exists y. Q(x,y)$.
This verifies that the specification is indeed satisfiable, i.e.,
that for every input that satisfies the precondition there exists some
output that satisfies the postcondition. If this theorem does not hold,
every attempt to prove the correctness of an algorithm with respect to
the specification is a priori doomed to fail.

\item Optionally, check the validity of $\forall x,y_1,y_2. P(x) \wedge Q(x,y_1)
\wedge Q(x,y_2) \Rightarrow y_1=y_2$. Thus we verify that the specification
defines the output uniquely; this, however, needs not generally be the case for
all specifications (algorithms are often intentionally underspecified).

\item Evaluate for all inputs the function
$f(x)\ \texttt{requires}\ P(x):=\texttt{choose}\ y: Q(x,y)$ which implicitly
defines its result by the postcondition (if the postcondition defines the result
uniquely, the evaluation may be performed in deterministic mode). By
inspecting the function results, we may validate that the specification
indeed describes the expected input/output behavior.

\end{enumerate}

Currently, the various theorems and the implicitly defined function have to be 
formulated manually by the user; their automatic generation by RISCAL is
planned in the near future.

\item \emph{The verification of algorithms:} this involves the definition of
procedures, their annotation with specifications, and (optionally) the
annotation of loops in the procedure bodies with invariants and termination
measures. The correctness of a procedure and the adequacy of its annotations
can be checked in RISCAL as follows:

\begin{enumerate}

\item Execute the procedure for all possible inputs. This verifies that for all
inputs that satisfy the precondition the procedure result indeed satisfies the
postcondition. If loop invariants and termination measures are given, this
also verifies that the invariants are not too strong (i.e., they hold before and
after every loop iteration) and that the termination measures are adequate (they
are decreased by every loop iteration and do not become negative). Thus, if a
termination measure is given, this guarantees the loop to terminate.

\item Derive from the procedure specification and the loop annotations
verification conditions that ensure the (partial or total) correctness of the
program with respect to the specification and check these. This in particular
verifies that the invariants are \enquote{inductive}, i.e, not too weak: from
the fact that the invariant holds before a loop iteration and the fact that the
loop condition holds, it can be concluded that the invariant also holds after
the iteration; it also verifies that from the loop invariant the postcondition
of the procedure can be concluded.

\end{enumerate}

Currently, the derivation of verification conditions has to be manually
performed by the application of Hoare calculus respectively Dijkstra's predicate
transformer calculus; their automatic generation in RISCAL is planned in the
near future.

\end{enumerate}

We illustrate some of above activities, concretely the specification and
verification of algorithms, by the following problem: given an array $a$ of
$n>0$ integers, find the maximum $m$ of $a$. The corresponding RISCAL
specification is based on the following domains:
\begin{xverbatim}
val N:ℕ; val M:ℕ;
type index = ℤ[-N,N]; type elem  = ℤ[-M,M]; type array = Array[N,elem];
\end{xverbatim}
Rather than basing our specifications on a single integer type, we use two
constants $N$ and $M$ to bound the types \texttt{index} of array
indices/lengths and the type \texttt{elem} of array elements, respectively.
The type \texttt{array} encompasses all arrays of length $N$; however,
the following specification uses a variable $n \leq N$ to denote the 
portion of the array actually considered in the problem. In the following
checks, we will use $N:=3$ and $M:=2$; thus arrays up to length 3
with values in the interval $[-2,2]$ will be considered.

The problem specification is then captured by the following predicates
representing the precondition and the postcondition of the problem, respectively:
\begin{xverbatim}
pred Pre(a:array, n:index) ⇔ 
  0 < n ∧ ∀k:index. n ≤ k ∧ k < N ⇒ a[k] = 0;
pred Post(a:array, n:index, m:elem) ⇔
  (∃k:index. 0 ≤ k ∧ k < n ∧ m = a[k]) ∧
  (∀k:index. 0 ≤ k ∧ k < n ⇒ m ≥ a[k]);
\end{xverbatim}
Here the precondition, in addition to requiring $n>0$, states that 
from index $n$ on all array elements are zero; while not strictly required,
this subsequently reduces the model space and thus speeds up all checks.

The specification is then validated with the help of the following declarations
that relate to the activities (2a)--(2e) mentioned above:
\begin{xverbatim}
theorem preSat ⇔ ∃a:array, n:index. Pre(a, n);
theorem postNotValid(a:array, n:index) ⇔ 
  Pre(a,n) ⇒ ∃m:elem. ¬Post(a,n,m);
theorem postSat(a:array, n:index) ⇔
  Pre(a,n) ⇒ ∃m:elem. Post(a,n,m);
theorem resultUnique(a:array, n:index, m1:elem, m2:elem) ⇔
  Pre(a,n) ∧ Post(a,n,m1) ∧ Post(a,n,m2) ⇒ m1 = m2;
fun maxFun(a:array, n:index): elem
  requires Pre(a,n);
= choose m:elem with Post(a,n,m);
\end{xverbatim}
Theorem \texttt{preSat} states that the precondition is satisfiable; since it
represents a constant, its value is immediately computed and checked when the
specification is processed. Theorems \texttt{postNotValid} (the
postcondition is not generally valid), \texttt{postSat} (the postcondition is
satisfiable), and \texttt{resultUnique} (the postcondition determines the result
uniquely) are individually verified by corresponding calls of the checker
(in non-deterministic mode with silent execution):
\begin{xverbatim}
Executing postNotValid(Array[ℤ],ℤ) with all 875 inputs.
Execution completed for ALL inputs (171 ms, 875 checked, 0 inadmissible).
Executing postSat(Array[ℤ],ℤ) with all 875 inputs.
Execution completed for ALL inputs (191 ms, 875 checked, 0 inadmissible).
Executing resultUnique(Array[ℤ],ℤ,ℤ,ℤ) with all 21875 inputs.
13638 inputs (13638 checked, 0 inadmissible, 0 ignored)... 
Execution completed for ALL inputs (3199 ms, 21875 checked, 0 inadmissible).
\end{xverbatim}
We further validate the specification by checking \texttt{maxFun}, now with 
non-silent execution in deterministic mode (to reduce the amount of output):
\begin{xverbatim}
Executing maxFun(Array[ℤ],ℤ) with all 875 inputs.
Ignoring inadmissible inputs...
Run 560 of deterministic function maxFun([-2,0,0],1):
Result (0 ms): -2
Run 561 of deterministic function maxFun([-1,0,0],1):
Result (0 ms): -1
...
Run 698 of deterministic function maxFun([1,2,0],2):
Result (0 ms): 2
...
Run 874 of deterministic function maxFun([2,2,2],3):
Result (0 ms): 2
Execution completed for ALL inputs (1146 ms, 155 checked, 720 inadmissible).
Not all nondeterministic branches may have been considered.
\end{xverbatim}
Having convinced ourselves about the adequacy of the specification, we
turn to the usual algorithm that solves the specified problem:
\begin{xverbatim}
proc maxProc(a:array, n:index): elem
  requires Pre(a,n);
  ensures  Post(a,n,result);
{
  var m:elem ≔ a[0]; 
  for var i:index≔1; i < n; i≔i+1 do
    invariant Invariant(a,n,m,i);
    decreases Termination(a,n,m,i);
  {
    if a[i] > m then m ≔ a[i];
  }
  return m;
}
\end{xverbatim}
The loop is annotated with the help of a predicate \texttt{Invariant}
and a function \texttt{Termination} that denote the invariant and
the termination term, respectively:
\begin{xverbatim}
pred Invariant(a:array, n:index, m:elem, i:index) ⇔
  1 ≤ i ∧ i ≤ n ∧
  (∃k:index. 0 ≤ k ∧ k < i ∧ m = a[k]) ∧
  (∀k:index. 0 ≤ k ∧ k < i ⇒ m ≥ a[k]);
fun Termination(a:array, n:index, m:elem, i:index):index = n-i;
\end{xverbatim}
By checking the procedure (activity 3a), we verify its correctness with respect 
to the specification, and (partially) validate the adequacy of invariant and 
termination term:
\begin{xverbatim}
Executing maxProc(Array[ℤ],ℤ) with all 875 inputs.
Execution completed for ALL inputs (93 ms, 155 checked, 720 inadmissible).
\end{xverbatim}
For a full validation of the adequacy of the termination term, we derive
the usual verification conditions for proving the total correctness
of the algorithm:
\begin{xverbatim}
theorem VC1(a:array, n:index, m:elem, i:index) 
  requires Pre(a,n);
⇔ m = a[0] ∧ i = 1 ⇒ Invariant(a,n,m,i);
theorem VC2(a:array, n:index, m:elem, i:index) 
  requires Pre(a,n);
⇔ Invariant(a,n,m,i) ⇒ Termination(a,n,m,i) ≥ 0;
theorem VC3(a:array, n:index, m:elem, i:index) 
  requires Pre(a,n);
⇔ Invariant(a,n,m,i) ∧ i < n ∧ a[i] > m ⇒ Invariant(a,n,a[i],i+1) ∧ Termination(a,n,a[i],i+1)≥0;
theorem VC4(a:array, n:index, m:elem, i:index)
  requires Pre(a,n);
⇔ Invariant(a,n,m,i) ∧ i < n ∧ ¬(a[i] > m) ⇒ Invariant(a,n,m,i+1) ∧ Termination(a,n,m,i+1) ≥ 0;
theorem VC5(a:array, n:index, m:elem, i:index)
  requires Pre(a,n);
⇔ Invariant(a,n,m,i) ∧ ¬(i < n) ⇒ Post(a,n,m);
\end{xverbatim}
All of these conditions are now checked in silent mode:
\begin{xverbatim}
Executing VC1(Array[ℤ],ℤ,ℤ,ℤ) with all 30625 inputs.
18714 inputs (3255 checked, 15459 inadmissible, 0 ignored)... 
Execution completed for ALL inputs (3188 ms, 5425 checked, 25200 inadmissible).
...
Executing VC5(Array[ℤ],ℤ,ℤ,ℤ) with all 30625 inputs.
21638 inputs (3725 checked, 17913 inadmissible, 0 ignored)... 
Execution completed for ALL inputs (2838 ms, 5425 checked, 25200 inadmissible).
\end{xverbatim}
Thus the algorithm is correct and the invariant is adequate for arrays of
lengths up to $N=3$ with absolute element values up to $M=2$. In order to verify
the correctness of the algorithm for arbitrary $N$ and $M$ we may pass above
conditions to a system that supports real (automated or interactive) reasoning
such as the RISC ProofNavigator~\cite{Schreiner2008}. If it can be shown that
above verification conditions hold for arbitrary~$N$ and~$M$, the algorithm is
indeed generally correct.

\section{Sample Specifications}
\label{sect:examples}

Since the release of the first version of RISCAL in March 2017, we have started
to develop first prototypes of formally checked specifications. They are
intended to serve as the nucleus of a future comprehensive library and
accompanying lecture materials to be used in the class room and for
self-instructed learning in degree programmes for computer science and
mathematics. The specifications include areas such as array-based algorithms,
logic, number theory, discrete mathematics, and computer algebra. Sample
specifications from two of these domains are described in somewhat more detail
below.

\subsection{Number Theory}
\label{sect:number}

In \cite{Fuerst2017}, we describe the application of RISCAL to number-theoretic
algorithms that arise in, e.g., cryptography. In such algorithms, \emph{prime
numbers} play an important role; thus as our first example we pick the problem
of generating, for a given bound $n\in\mathbb{N}$, all prime numbers less than
equal $n$: formally, we wish to compute every $p\in\mathbb{N}$ with $p \leq n$
that satisfies the following predicate $\mathit{isPrime}(p)$:
\begin{align*}
\mathit{isPrime}(p):\Leftrightarrow p > 1 \wedge \forall n\in\mathbb{N}.\
n | p \Rightarrow n=1 \vee n=p
\end{align*}
Here the predicate $n|p$ (\enquote{$n$ divides $p$}) is defined as usual as $n|p
:\Leftrightarrow \exists m\in\mathbb{N}.\ n\cdot m = p$.

The corresponding declarations in RISCAL are now as follows:
\begin{xverbatim}
val N: ℕ; type nat = ℕ[N];
pred divides(n:nat,p:nat) ⇔ ∃m:nat. n⋅m = p;
pred isPrime(p:nat) ⇔ p > 1 ∧ ∀n:nat. divides(n,p) ⇒ n = 1 ∨ n = p;
\end{xverbatim}
Here we first introduce the type \texttt{nat} of all natural numbers up to
some maximum $N$ and then define the predicates \texttt{divides(m,n)} 
representing $m|n$ and \texttt{isPrime(p)} in the natural way.

The classic algorithm for the solution of this problem is the \enquote{Sieve
of Eratosthenes}. From the educational point of view, this algorithm has the
advantage that it is easily understandable even to high school students and,
furthermore, that it gives us the opportunity to present two variants of an
algorithm: a \enquote{mathematical} one that is described on a higher level of
abstraction (by operating on sets), and a \enquote{computational} one that is
more implementation oriented (by operating on arrays).

In both variants, the algorithm is based on the following fundamental knowledge.
\begin{theorem}[Least Proper Divisor]
Let $n$ be a natural number greater equal 2. Then the least proper divisor 
$m> 1$ of $n$ is a prime.
\end{theorem}
\begin{proof}
Let $m> 1$ be the least proper divisor of $n$. We assume that $m$ is not
a prime and show a contradiction. Since $m$ is a divisor of $n$, there exists
some $m'\in\mathbb{N}$ with $n=m\cdot m'$. Since $m$ is not a prime, there exists
some $o\in\mathbb{N}$ with $o\neq 1$ and $o\neq m$ such that $o|m$.
Therefore there also exists some $o'\in\mathbb{N}$ with $o'\neq 1$ and $o'\neq m$
and $m=o\cdot o'$. Therefore, we have $n=m\cdot m'=o\cdot o'\cdot m'$. But then
$o<m$ is also a divisor of $n$, i.e., $m$ is not the least divisor, which
contradicts our assumption.
\end{proof}

While this argument can be easily understood by students, we can give
the theorem further credibility by formulating it in RISCAL:
\begin{xverbatim}
theorem leastProperDivisor(n:nat,m:nat) ⇔
  n ≥ 2 ∧ m > 1 ∧ divides(m,n) ∧ (∀m0:nat. m0 > 1 ∧ divides(m0,n) ⇒ m ≤ m0)
  ⇒ isPrime(m);
\end{xverbatim}
Now checking the theorem with for example $N=30$ yields
\begin{xverbatim}
Executing leastProperDivisor(ℤ,ℤ) with all 961 inputs.
Execution completed for ALL inputs (95 ms, 961 checked, 0 inadmissible).
\end{xverbatim}
which validates its correctness even without proof.

This theorem forms the basis of the Sieve of Eratosthenes: starting with the
prime number candidates~$2,\ldots,n$, we repeatedly pick the smallest value
(which by above theorem is a prime) and remove it together with its multiples
(which by definition are not primes); this process is repeated until no more
candidate is left. In the following, we give two formulations of this algorithm.

\paragraph{Sieve of Eratosthenes (Set-Based)}

The following algorithm computes for given $n\in\mathbb{N}$ the set $P$
of all values $p\leq n$ with $\mathit{isPrime}(p,n)$:

\begin{algorithm}[H]
\begin{algorithmic}[1]
%\Require{$n\in\mathbb{N}$}
\Ensure{$P = \{ p\ |\ p\in\mathbb{N} \wedge p\leq n \wedge \mathit{isPrime}(p) \}$}
%\Procedure{SieveOfEratosthenesSet}{$n$}
%\Comment{ensures $\forall e\in\mathbb{N}.\ e\leq n \Rightarrow (e\in P
%\Leftrightarrow\mathit{isPrime}(e))$}
\State $P \gets \emptyset$
\State $C \gets \{2,\ldots,n\}$
\While{$C \neq \emptyset$}
\State $p \gets \min(C)$
\State $P\gets P \cup \{p\}$
\State $C\gets \{c\in C: p \not|\ c\}$
\EndWhile
%\State\Return $P$
%\EndProcedure
\end{algorithmic}
%\end{quote}
        \caption[caption]{%SieveOfEratosthenesSet($n$,$P$): %\\\hspace{\textwidth}
        Compute the set $P$ of all primes less than equal $n\in\mathbb{N}$.}
        \label{alg:2}
\end{algorithm}

\noindent
This algorithm can be formalized in RISCAL (in the most elegant way) as follows:
\begin{xverbatim}
proc SieveOfEratosthenesSet(n:nat): Set[nat]
  ensures result = { p | p:nat with p ≤ n ∧ isPrime(p) };
{
  var P:Set[nat] ≔ ∅[nat];
  var C:Set[nat] ≔ 2..n;
  choose p∈C with ∀c∈C. p ≤ c do
  {
    P ≔ P ∪ {p};
    C ≔ { c | c∈C with ¬divides(p,c) };
  }
  return P;
}
\end{xverbatim}
The construct \texttt{choose \ldots\ do \ldots} repeatedly chooses the minimum
value $p$ from the set $C$ of candidates, puts it as the next prime number into
set $P$, and removes it together with its multiples from $C$; the process
terminates, when no more choice is possible (i.e., when $C$ is empty).

A general proof of correctness of the algorithm can be based on the
following loop annotations:
\begin{xverbatim}
invariant P = { p | p:nat with p ≤ n ∧ isPrime(p) ∧ ∀c∈C. p < c };
invariant C ⊆ 2..n;
decreases |C|;
\end{xverbatim}
Here the invariants demonstrate that $P$ contains, at every iteration of the
loop, all primes that are less than equal $n$ and less than equal the minimum of
$C$. Thus, when the loop terminates with $C=\emptyset$, $P$ holds all primes
less than equal $n$. Since the termination term $|C|$ indicates that the size
of~$C$ is decreased in every loop iteration, this is indeed eventually the case.
Checking the algorithm with $N=30$
\begin{xverbatim}
Executing SieveOfEratosthenesSet(ℤ) with all 31 inputs.
Execution completed for ALL inputs (746 ms, 31 checked, 0 inadmissible).
\end{xverbatim}
validates both that the algorithm satisfies its contract and that the
loop annotations are correct (if they should not be sufficient for a proof, 
they are at least not too strong).

\paragraph{Sieve of Eratosthenes (Array-Based)}

The following algorithm computes for given $n\in\mathbb{N}$ the Boolean
array~$P$ of length $n$ such that $P[p]$ has value \enquote{true} if and only if
the property $\mathit{isPrime}(p)$ holds:

\begin{algorithm}[H]
\begin{algorithmic}[1]
%\Require{$n\in\mathbb{N}$ such that $n\geq 2$}
\Ensure{$\forall p\in\mathbb{N}.\ p\leq n \Rightarrow (P[p] = \textrm{T}
\Leftrightarrow\mathit{isPrime}(p))$}
\State $P \gets ({\rm T},{\rm T},\ldots,{\rm T}) \in \mathbb{B}^{n+1}$;
\State $P[0],P[1] \gets {\rm F}$;
\For{$p$\, \textbf{from} $2$ \textbf{while} $p\cdot p \leq n$ \textbf{by} $1$}
  \If{$P[p]={\rm T}$}
  \For{$k$ \textbf{from} 2 \textbf{while} $p\cdot k \leq n$ \textbf{by} $1$}
	  \State $P[p\cdot k] \gets {\rm F}$
  \EndFor
  \EndIf
\EndFor
%\State\Return $P$.
\end{algorithmic}
        \caption[caption]{%SieveOfEratosthenesArray($n$,$P$): %\\\hspace{\textwidth}
        Compute the Boolean array $P$ which indicates all primes less than equal
        $n\in\mathbb{N}$.}
        \label{alg:3}
\end{algorithm}

This algorithm is in essence a refinement of the set-based algorithm, where
the Boolean array $P$ takes the role of both sets $P$ and $C$: all array values at
indices less than $p$ already indicate the prime status of the indices 
while all indices greater than equal $p$ are the candidates that remain to
be processed. The outer loop looks for the next smallest prime number $p$; when
such a $p$ is found, its greater multiples are removed from the candidates. 
The algorithm can be formalized in RISCAL as follows:

\begin{xverbatim}
proc SieveOfEratosthenesArray(n:nat): Array[N+1,Bool]
  ensures ∀p:nat with p ≤ n. result[p] ⇔ isPrime(p);
{
  var P:Array[N,Bool] ≔ Array[N+1,Bool](⊤);
  P[0] ≔ ⊥; P[1] ≔ ⊥;
  for var p:nat ≔ 2; p⋅p ≤ n; p ≔ p+1 do
  {
    if P[p] then
    {
      for var k:nat ≔ 2; p⋅k ≤ N; k ≔ k+1 do
        P[p⋅k] ≔ ⊥;
    }
  }
  return P;
}
\end{xverbatim}

A verification of the algorithm can be based on the
following annotations of the outer loop:
\begin{xverbatim}
invariant 2 ≤ p ∧ (p-1)⋅(p-1) ≤ n;
invariant ∀j:nat with j < p. P[j] ⇔ isPrime(j);
invariant ∀j:nat with 2 ≤ j ∧ j < p. ∀k:nat with j < k. divides(j,k) ⇒ ¬P[k];
decreases n-p+2;
\end{xverbatim}
The invariants essentially state that the prime status for all positions less than $p$
has been already determined and all positions from $p$ on do not hold multiples
of the smaller values. Correspondingly we have the 
following annotations of the inner loop:
\begin{xverbatim}
invariant 2 ≤ k ∧ (p-1)⋅k ≤ N;
invariant ∀j:nat with 2 ≤ j ∧ j < k. ¬P[p⋅j];
decreases N-k;
\end{xverbatim}
Here the invariants essentially state that all multiples $p\cdot j$ with
$j < k$ have already removed as prime number candidates.
Checking also this algorithm with $N=30$
\begin{xverbatim}
Executing SieveOfEratosthenesArray(ℤ) with all 31 inputs.
Execution completed for ALL inputs (171 ms, 31 checked, 0 inadmissible).
\end{xverbatim}
again validates both that the algorithm satisfies its contract and that the
loop annotations are correct (if they should not be sufficient for a proof, 
	they are at least not too strong).

It is generally a challenging task to come up with completely correct
formalizations that in particular also hold at the boundaries of iteration
respectively quantification ranges. From an educational point of view, failures
reported by the checker are a valuable way to learn to take extra care of
boundary conditions and special cases; these situations may be easily overlooked
by a human but are also easily handled by automatic checking.

%It was the third author's personal joy to
%experience this first hand when working on~\cite{Fuerst2017}.

\subsection{Discrete Mathematics}
\label{sect:discrete}

In \cite{Brunhuemer2017}, chosen theories from discrete mathematics have been
formalized in RISCAL. This includes the development of the mathematical
theories, the specification of algorithms and their annotation with
meta-information. The concepts have been validated on small finite domains which
should work as a ground layer for further verification on models of arbitrary
size. This approach should save much time by quickly finding errors in
considerations, because in most cases errors in the specifications and
annotations from the generalized concepts also appear in the finite domains.

A big focus in the elaborations lies on drawing the connections between
different ways of describing an algorithm/function, which leads to a deeper
understanding of the underlying theories. Most functions we have defined in
three variants: \emph{implicitly} by a condition that the result shall fulfill
(by a RISCAL term of form \texttt{choose I:T with P}, which denotes a value
\texttt{I} of type \texttt{T} with property \texttt{P}), \emph{explicitly} by a
constructive (generally recursive) description how to find the result (which
requires a termination measure to ensure the well-definedness of the
definition), and finally \emph{procedurally} by the execution of a sequence of
commands that update the values of variables.

We demonstrate these connections by the problem of computing the
\emph{transitive closure} of a
binary relation based on the following mathematical concepts.
\begin{definition}[Relation] $R$ is a
\emph{binary relation on $E$} if it is a set of pairs of elements of $E$, 
i.e., $R \subseteq E \times E$. If $E$ is fixed from the context, then
we call $R$ just a \emph{relation}.
\end{definition}
\noindent
We define correspondingly in RISCAL the type of binary relations on the set 
of numbers $\{0,\ldots,N\}$ for some maximum $N\in\mathbb{N}$:
\begin{xverbatim}
val N:ℕ;
type elem = ℕ[N];
type pair = Tuple[elem,elem];
type relation = Set[pair];
\end{xverbatim}
\begin{definition}[Transitivity] A binary relation $R$ on $E$ is
\textit{transitive} if, whenever $\langle x,y \rangle \in
R$ and $\langle y,z \rangle \in R$, then $\langle x,z \rangle
\in R$, for all $x,y,z \in E$.
\end{definition}
\noindent
In RISCAL we define this property by the following predicate:
\begin{xverbatim}
pred isTransitive(r:relation) ⇔ ∀x∈r,y∈r. (x.2 = y.1) ⇒ ⟨x.1,y.2⟩ ∈ r;
\end{xverbatim}
Based on these fundamental definitions we can now specify the
central problem of this section.

\paragraph{Problem Specification and Implicit Definition}

Our problem is to compute the transitive closure of a given binary relation
based on the following definition.
\begin{definition}[Transitive Closure] $S$ is the
\emph{transitive closure} of relation $R$, if $S$ is the smallest 
transitive relation $S$ that contains $R$, i.e., $S$ is a subset of every 
transitive relation that contains $R$.
\end{definition}
\noindent
The corresponding RISCAL definition is:
\begin{xverbatim}
pred isTransitiveClosure(r:relation,s:relation) ⇔ 
  r ⊆ s ∧ isTransitive(s) ∧ 
  (∀s0:relation. r ⊆ s0 ∧ isTransitive(s0) ⇒ s ⊆ s0);
\end{xverbatim}
We claim that the transitive closure of every relation indeed
exists and, furthermore, is uniquely defined:
\begin{xverbatim}
theorem transitiveClosureExists(r:relation) ⇔ 
  ∃s:relation. isTransitiveClosure(r,s);
theorem transitiveClosureIsUnique(r:relation) ⇔ 
  ∀s1:relation with isTransitiveClosure(r,s1).
  ∀s2:relation with isTransitiveClosure(r,s2).
    s1 = s2;
\end{xverbatim}
These properties can be quickly validated for $N=2$:
\begin{xverbatim}
Executing transitiveClosureExists(Set[Tuple[ℤ,ℤ]]) with all 512 inputs.
Execution completed for ALL inputs (675 ms, 512 checked, 0 inadmissible).
Executing transitiveClosureIsUnique(Set[Tuple[ℤ,ℤ]]) with all 512 inputs.
362 inputs (362 checked, 0 inadmissible, 0 ignored)... 
Execution completed for ALL inputs (2648 ms, 512 checked, 0 inadmissible).
\end{xverbatim}
Thus we can implicitly define a function which chooses a relation
that satisfies the required property:
\begin{xverbatim}
fun transitiveClosureI(r:relation):relation = 
  choose s:relation with isTransitiveClosure(r,s);
\end{xverbatim}
To validate our definitions, we can execute this function for all possible
inputs and inspect its results (because the result is uniquely defined,
deterministic execution suffices):
\begin{xverbatim}
Executing transitiveClosureI(Set[Tuple[ℤ,ℤ]]) with all 512 inputs.
Run 0 of deterministic function transitiveClosureI({}):
Result (3 ms): {}
...
Run 106 of deterministic function transitiveClosureI({[1,0],[0,1],[2,1],[0,2]}):
Result (1 ms): {[0,0],[1,0],[2,0],[0,1],[1,1],[2,1],[0,2],[1,2],[2,2]}
...
Execution completed for ALL inputs (4242 ms, 512 checked, 0 inadmissible).
Not all nondeterministic branches may have been considered.
\end{xverbatim}

\paragraph{Explicit (Recursive) Definition}

A natural approach to compute the transitive closure of a relation $r$
as a recursive function is based on the following fundamental steps:
\begin{enumerate}
\item If $r$ is transitive, then we are finished and $r$ is our result.
\item Otherwise $r$ contains some pairs $\langle x,z \rangle$ and
$\langle z,y \rangle$ that violate the transitivity of $r$ in the sense
that $\langle x,y \rangle$ is not in $r$. We thus add $\langle x,y \rangle$ to 
$r$ and continue with step (1).
\end{enumerate}
Indeed, in step (2) of the algorithm we may consider all violating pairs and add
to $r$ all the mending pairs at once. The corresponding function can be
defined in RISCAL as follows:
\begin{xverbatim}
fun transitiveClosureR(r:relation):relation
  ensures isTransitiveClosure(r,result);
  decreases 2^((N+1)^2)-|r|;
= if isTransitive(r) then 
    r
  else 
    transitiveClosureR(r ∪ 
      { ⟨x,y⟩ | x:elem,y:elem with ∃p∈r,q∈r. (x = p.1 ∧ y = q.2 ∧ p.2 = q.1) });
\end{xverbatim}
The termination of this function is guaranteed by the measure stated
in the \texttt{decrease} clause; its correctness follows from the fact that
the size $|r|$ of $r$ is increased in every recursive invocation; however, since
we have only $N+1$ elements, there are at most $(N+1)^2$ pairs in $r$, thus
$|r|$ can be at most $2^{(N+1)^2}$.

The partial correctness of this algorithm is a direct consequence of the
following theorem on which a later verification may be based:
\begin{xverbatim}
theorem transitiveClosureCorrectness(r:relation) ⇔
  if isTransitive(r) then
    isTransitiveClosure(r,r)
  else
    let s = r ∪ { ⟨x,y⟩ | x:elem,y:elem with ∃p∈r,q∈r. (x = p.1 ∧ y = q.2 ∧ p.2 = q.1) } in
    ∀t:relation. isTransitiveClosure(r ∪ s,t) ⇒ isTransitiveClosure(r,t);
\end{xverbatim}
Both the algorithm and the correctness theorem may be quickly validated for $N=2$:
\begin{xverbatim}
Executing transitiveClosureR(Set[Tuple[ℤ,ℤ]]) with all 512 inputs.
Execution completed for ALL inputs (699 ms, 512 checked, 0 inadmissible).
Executing transitiveClosureCorrectness(Set[Tuple[ℤ,ℤ]]) with all 512 inputs.
Execution completed for ALL inputs (1148 ms, 512 checked, 0 inadmissible).
\end{xverbatim}

\paragraph{Procedural Definition}

Another algorithm which may be easiest expressed as a procedure is based on the
following main steps:
\begin{enumerate}
\item We initialize the result variable $\mathit{res}$ with the empty set and
an auxiliary variable $\mathit{new}$ with $r$. 

\item We choose some pair $x\in\mathit{new}$ and check for every
pair $y\in\mathit{res}$, if the combination of $x$ and $y$ violates the
transitivity of $\mathit{res}$. If yes, add to $\mathit{new}$ the pair that mends
the violation.

\item We add $x$ to $\mathit{res}$, remove it from $\mathit{new}$ and
continue with step (2).

\item When $\mathit{new}$ becomes empty, the algorithm terminates
and we return $\mathit{res}$ as its result.
\end{enumerate}
In more detail, the algorithm can be formulated in RISCAL as follows:
\begin{xverbatim}
proc transitiveClosureP(r:relation):relation
  ensures isTransitiveClosure(r,result);
{
  var res:relation ≔ ∅[pair];
  var new:relation ≔ r;
  choose x ∈ new do
  {
    for y ∈ res do
    {
      if x.1 = y.2 ∧ ¬(⟨y.1,x.2⟩ ∈ res) then 
        new ≔ new ∪ { ⟨y.1, x.2⟩ };
      if x.2 = y.1 ∧ ¬(⟨x.1,y.2⟩ ∈ res) then
        new ≔ new ∪ { ⟨x.1,y.2⟩ };
    }
    res ≔ res ∪ { x };
    new ≔ new \ { x };
  }
  return res;
}
\end{xverbatim}
The termination of the algorithm can
be guaranteed by adding the termination measure
\begin{xverbatim}
decreases 2^((N+1)^2)-|res|;
\end{xverbatim}
to the outer loop. Similar to the recursive algorithm, its correctness follows 
from the fact that the size of $\mathit{res}$ is increases by every loop
iteration but cannot exceed $2^{(N+1)^2}$.

The partial correctness of the algorithm follows from the following
invariants of the outer loop:
\begin{xverbatim}
invariant res ∩ new = ∅[pair]; 
invariant res ∪ new ⊆ transitiveClosureI(r);  
invariant ∀s∈res,t∈res with s.2 = t.1. ⟨s.1,t.2⟩ ∈ res ∨ ⟨s.1,t.2⟩ ∈ new;
\end{xverbatim}
From the second invariant, we know that in every loop iteration $\mathit{res}$ 
is a subset of the transitive closure. Since the loop terminates when
$\mathit{new}$ is the empty set, the third invariant implies that then
$\mathit{res}$ is transitive and thus itself the transitive closure.

The correctness of the invariant of the outer proof has to be verified
with the help of the invariant of the inner loop:
\begin{xverbatim}
invariant new = old_new 
  ∪ { ⟨y0.1,x.2⟩ | y0 ∈ forSet with x.1 = y0.2 ∧ ¬⟨y0.1,x.2⟩ ∈ res }
  ∪ { ⟨x.1,y0.2⟩ | y0 ∈ forSet with x.2 = y0.1 ∧ ¬⟨x.1,y0.2⟩ ∈ res };
\end{xverbatim}
This invariant explicitly describes which values have been added after the
termination of the loop to the original value of $\mathit{new}$ (special
variable $\mathit{old\_new}$) to get the new value.

The correctness of the algorithm with respect to its specification and
the correctness of the annotations may be quickly validated for $N=2$:
\begin{xverbatim}
Executing transitiveClosureP(Set[Tuple[ℤ,ℤ]]) with all 512 inputs.
Execution completed for ALL inputs (1217 ms, 512 checked, 0 inadmissible).
Not all nondeterministic branches may have been considered.
\end{xverbatim}
Both loops have been expressed by non-deterministic iteration constructs
\texttt{choose} and \texttt{for} which may select the respective elements
in arbitrary order; here deterministic execution has been selected 
in order to avoid the combinatorial explosion of execution branches.

With the help of the RISCAL support for parallelism, it is also possible
to check larger domain instances. For instance, we may check the
instance $N=3$ with 4 threads on the local computer and 16 threads
on a remote server in a considerable but still manageable amount of time:
\begin{xverbatim}
Executing transitiveClosureP(Set[Tuple[ℤ,ℤ]]) with all 65536 inputs.
...
PARALLEL execution with 4 local threads and 1 remote servers (output disabled).
...
Execution completed for ALL inputs (870524 ms, 65536 checked, 0 inadmissible).
Not all nondeterministic branches may have been considered.
\end{xverbatim}
However, as emphasized before, the point of RISCAL is not so much in verifying
specifications but finding errors; typically such errors already arise
in small domain instances that allow quicker checking.

%\paragraph{Checking}
%The RISCAL environment is a powerful tool, when it comes to
%validation of specifications. By stating suitable
%\emph{postconditions} it was very simple to assure, that the
%yielded result suffices the desired properties (correctness of the algorithms).
%emph{Invariants} (as well as \emph{preconditions} for recursive
%functions) lead to a deep understanding of every single
%iteration (or function call), and are a very convenient way to
%find errors in one's considerations. Additionally, the
%\emph{termination measures} allow to verify if our algorithm
%ends after a finite number of steps or iterations (termination). When errors
%occur in the definitions of the annotations, they often can be
%revealed, by running the model checks on the specification. This
%really saves one's nerves, since finding errors in wrong
%declared conditions is without doubt absolutely frustrating.
%All in all the checking by the RISCAL environment allows very
%straight thinking towards a solution, because after every
%execution of the model checking we know, which conditions
%have been violated, and it becomes possible to intervene at the
%right spots.

\section{Conclusions and Further Work}
\label{sect:conclusions}

RISCAL is still in its infancy with first contents
having been developed and first classroom experience having been gained;
subsequent experience with the use of RISCAL will shape the further
development of language and system. This will proceed along several
possible strands:
\begin{itemize}
\item We will automatize the generation of formulas for the validation
of specifications and of verification conditions 
for the proof-based verification of algorithms over arbitrary size domains.
\item As a (potentially more efficient) alternative to formula evaluation, we 
will translate formulas into a suitable theory of the
SMT-LIB library and apply SMT solvers for checking their validity.
\item We will investigate the visualization of formula evaluation in order
to give students quick feedback why a formula is not valid.
\item We will export formulas to some external prover(s) such as the 
RISC ProofNavigator to allow the seamless proof-based verification of formulas
over arbitrary size domains.
\end{itemize}
Most important, however, we want to develop a library of specifications and
accommodating lecture materials that shall support the self-study of students
in the formalization of mathematical theories and algorithms with the ultimate
goal of self-paced and self-instructed learning.

\bibliographystyle{eptcs}
\bibliography{main}

\end{document}